\documentclass[conference,a4paper]{IEEEtran}

\usepackage{amsfonts}
\usepackage{amsmath,amsthm}
\usepackage{amssymb}
\usepackage{graphicx,color}

\DeclareMathAlphabet      {\mathbfit}{OML}{cmm}{b}{it} 

\newtheorem{theorem}{Theorem}
\newtheorem{lemma}{Lemma}
\newtheorem{EXAMPLE}{Example}
\newenvironment{example}{\begin{EXAMPLE}\rm}{\rm\end{EXAMPLE}} 

\begin{document}

\newcommand{\sff} {\mathsf{f}}
\newcommand{\code}{{\mathcal{C}}}

\sloppy

\title{Minimum Distance Distribution of Irregular Generalized LDPC Code Ensembles}

%
\author{
  \IEEEauthorblockN{Ian P.~Mulholland and Mark F.~Flanagan}
  \IEEEauthorblockA{School of Electrical, Electronic\\ \& Communications Engineering\\
       University College Dublin\\
       Dublin, Ireland\\
    Email: \{mulholland, mark.flanagan\}@ieee.org} 
  \and
  \IEEEauthorblockN{Enrico Paolini}
  \IEEEauthorblockA{Department of Electrical, Electronic,\\ 
  and Information Engineering ``G. Marconi''\\
       University of Bologna\\
       Cesena (FC), Italy\\
    Email: e.paolini@unibo.it}
}
 
\maketitle

\begin{abstract}
In this paper, the minimum distance distribution of irregular generalized LDPC (GLDPC) code ensembles is investigated. Two classes of GLDPC code ensembles are analyzed; in one case, the Tanner graph is regular from the variable node perspective, and in the other case the Tanner graph is completely unstructured and irregular. In particular, for the former ensemble class we determine exactly which ensembles have minimum distance growing linearly with the block length with probability approaching unity with increasing block length. This work extends previous results concerning LDPC and regular GLDPC codes to the case where a hybrid mixture of check node types is used. 
\end{abstract}

\section{Introduction}
Recently, the design and analysis of coding schemes representing generalizations of Gallager's low-density parity-check (LDPC) codes \cite{Gallager1963} has gained increasing attention. This interest is motivated above all by the potential capability of this class of codes to offer a better compromise between waterfall and error floor performance than is currently offered by state-of-the-art LDPC codes.

In the Tanner graph of an LDPC code, any degree-$s$ check node (CN) may be interpreted as a length-$s$ single parity-check (SPC) code, i.e., as an $(s,s-1)$ linear block code. The first proposal of a class of linear block codes generalizing LDPC codes may be found in \cite{Tanner1981}, where it was suggested to replace each CN of a regular LDPC code with a generic linear block code, to enhance the overall minimum distance. The corresponding coding scheme is known as a regular generalized LDPC (GLDPC) code, or Tanner code, and a CN that is not a SPC code as a generalized CN. More recently, irregular GLDPC codes were considered (see for instance \cite{Liva2008}). For such codes, the variable nodes (VNs) may exhibit different degrees and the CN set is composed of a mixture of different linear block codes.

In this paper, we present results on the minimum distance distribution of two classes of GLDPC code ensembles. It is shown that for the considered VN-regular ensembles, the ensembles for which the minimum distance grows linearly with the block length with probability approaching unity (with increasing block length) are precisely those which have \emph{good growth rate behavior} as defined in \cite{Flanagan2012}. For the unstructured irregular GLDPC ensembles, we provide an upper bound on the probability of the minimum distance lying below a certain fraction of the code's block length.

\section{Preliminaries and Notation}
In this work, we will consider two GLDPC code ensembles. These ensembles share definitions from the CN perspective, so we begin by giving these definitions. 

We define a GLDPC code ensemble as follows. There are $n_c$ different CN types $t\in I_c = \{1,2,\dots,n_c\}$. For each CN type $t\in I_c$ we associate a local code denoted by $\code_t$, and we denote by $k_t, s_t$ and $r_t$, the dimension, length and minimum distance of $\code_t$, respectively. For $t \in I_c$, $\rho_t$ denotes the fraction of edges connected to CNs of type $t$. The polynomial $\rho(x)$ is defined by $\rho(x) \triangleq \sum_{t\in I_c}\rho_t x^{s_t-1}$. 

If $E$ denotes the number of edges in the Tanner graph, the number of CNs of type $t\in I_c$ is then given by $E \rho_t / s_t$. Denoting as usual $\int_0^1 \rho(x) \, {\rm d} x$ by $\int \rho$, it is easily deduced that the number of CNs is given by $m = E \int \rho$. Therefore, the fraction of CNs of type $t \in I_c$ is given by
\begin{equation}
\label{eq:gamma_t_definition}
\gamma_t = \frac{\rho_t}{s_t \int \rho}\; .
\end{equation}
The parity-check matrix for CN type $t \in I_c$ is denoted by $\bar{\mathbf{H}}_t$. The weight enumerating function (WEF) for CN type $t \in I_c$ is given by
\begin{align*}
A^{(t)}(z) &= \sum_{u=0}^{s_t} A_u^{(t)} z^u = 1 + \sum_{u=r_t}^{s_t} A_u^{(t)} z^u \; .
\end{align*}
Here $A_u^{(t)} \ge 0$ denotes the number of weight-$u$ codewords for CNs of type $t$. We assume that the local codes associated with all CNs have minimum distance of at least $2$ (i.e., $r_t\geq 2$ for $t \in I_c$).

For ensembles which have a positive fraction of CNs with minimum distance $2$, the parameter $C$ is defined by
\begin{equation}
\label{eq:C_defn}
C = 2\sum_{t: r_t=2}\frac{\rho_t A_2^{(t)}}{s_t}.
\end{equation}

The number of VNs, which is also equal to the overall block length of the ensemble, is denoted by $N$. The two ensembles differ from the perspectives of VN distribution and Tanner graph interconnectivity. We next provide the further definitions for these two ensembles separately. 

\subsection{Ensemble 1}
\label{Ensemble1def}
Ensemble 1 is an extension of the definition given in \cite{Boutros1999} and in \cite{Lentmaier1997,Lentmaier1999} for regular GLDPC codes to the hybrid CN case (a related class of codes was also considered in \cite{Barg2008}). The overall parity-check matrix of the code is a formed by vertically concatenating $q \ge 2$ block rows $\mathbf{H}_{\ell}$, $\ell = 1,2,\ldots,q$. The first block row $\mathbf{H}_1$ is a block-diagonal matrix, whose diagonal elements consist of $\gamma_t m / q$ matrices $\bar{\mathbf{H}}_t$ for each $t \in I_c$. These are the parity-check matrices of the constituent codes associated with the $n_c$ CN types. Each of these parity-check matrices is repeated $\gamma_t m / q$ times along the diagonal. The resulting matrix $\mathbf{H_1}$, which forms the first of the $q$ block rows of the parity-check matrix for the GLDPC code, is given by 
\begin{equation*}
\mathbf{H}_1 =
\begin{pmatrix}
  \bar{\mathbf{H}}_1 & \mathbf{0}        &  \cdots             &                     & \cdots                 & \mathbf{0}             & \mathbf{0} \\
  \mathbf{0}         & \ddots            &                     &                     & \cdots                 & \mathbf{0}             & \mathbf{0} \\
  \vdots             &                   & \bar{\mathbf{H}}_1  &                     &                        & \vdots                 & \vdots     \\
                     &                   &                     & \ddots              &                        &                        &            \\            
  \vdots             & \vdots            &                     &                     & \bar{\mathbf{H}}_{n_c} &                        & \vdots     \\            
  \mathbf{0}         & \mathbf{0}        & \cdots              &                     &                        & \ddots                 & \mathbf{0}  \\            
  \mathbf{0}         & \mathbf{0}        & \cdots              &                     & \cdots                 & \mathbf{0}             & \bar{\mathbf{H}}_{n_c}      
\end{pmatrix}
\end{equation*}
The other $q-1$ block rows $\mathbf{H}_2,\dots,\mathbf{H}_q$ are formed by performing random column permutations $\Pi_2, \Pi_2, \ldots , \Pi_q$ on $\mathbf{H}_1$. Stacking the block-rows on top of one another results in $\mathbf{H}$, the parity-check matrix of the GLDPC code.

The ensemble is defined according to a uniform probability distribution on all permutations $\Pi_{\ell}$, for every $\ell = 2,3,\ldots,q$ (together with independence of these permutations). Note that the Tanner graph for this ensemble is \emph{VN-regular}, i.e., all VNs have the same degree $q$. The design rate $R$ for the irregular GLDPC codes in Ensemble 1 is given by
\begin{equation}
\label{GLDPC_code_rate_eqn}
R = 1-q\left(1 - \sum_{t\in I_c} \frac{\rho_t k_t}{s_t} \right).
\end{equation}

\subsection{Ensemble 2}
\label{Ensemble2def}
This second ensemble we consider is a generalization of the unstructured irregular LDPC ensemble analyzed in \cite{Richardson2001}, and is also a special case of the unstructured irregular doubly-generalized LDPC ensemble analyzed in \cite{Flanagan2011,Flanagan2012}. Here $\lambda_d$ denotes the fraction of edges connected to VNs of degree $d$, where $d\in\{2,3,\dots,d_v\}$. The polynomial $\lambda(x)$ is defined by $\lambda(x) \triangleq \sum_{d=2}^{d_v}\lambda_d x^{d-1}$. We denote as usual $\int^1_0\lambda(x){\rm d}x$ by $\int\lambda$. The \emph{node-perspective} VN degree distribution is defined as
\begin{equation}
\tilde{\lambda}_d=\frac{\lambda_d}{d \int\lambda}.
\label{eq:node_perspec_DD}
\end{equation}
Here $\tilde{\lambda}_d$ is the fraction of VNs having degree $d$.

The ensemble is defined according to a uniform probability distribution on all permutations connecting the $E$ edges of the Tanner graph. Note that whereas Ensemble 1 is VN-regular, Ensemble 2 is \emph{VN-irregular} (i.e., in general, the VNs do not all have the same degree). 

\section{Minimum Distance Results for Ensemble 1}
In \cite{Lentmaier1997,Lentmaier1999}, a lower bound on the minimum distance of a regular GLDPC code was found (generalizing the corresponding result for LDPC codes in \cite{Gallager1963}). Following a similar approach yields the following theorem in the case of Ensemble 1. 

\medskip
\begin{theorem}\label{mindistthm1}
Let $d_{\min}$ be the minimum distance of a GLDPC code picked randomly with uniform probability from Ensemble~1 described above. Then
\begin{equation}
\mathrm{Pr}(d_\mathrm{min}\leq \alpha^* N) \rightarrow 0 \quad \mathrm{as} \quad N \rightarrow \infty 
\end{equation}
where $\alpha^*$ is the smallest solution in the interval $(0,1)$ to the equation $G(\alpha) = 0$,

\begin{align}\label{eq:G(alpha)_irregular_CN_set}
G(\alpha) & = (1-q) h(\alpha) -q\, \alpha \log \sff^{-1}(\alpha) \notag \\ 
\, & \phantom{------}+ q \left(\smallint\!\rho\right) \sum_{t\in I_c} \gamma_t \log A^{(t)}(\sff^{-1}(\alpha)) \; ,
\end{align}
and the invertible function $\sff$ is given by
\begin{align}\label{eq:f}
\sff(z) = \left( \smallint\! \rho \right) \sum_{t \in I_c} \gamma_t \frac{ z \, \frac{\mathrm{d}
A^{(t)}(z)}{\mathrm{d} z}}{A^{(t)}(z)} \, .
\end{align}
Here $h(\alpha) = -\alpha\log(\alpha)-(1-\alpha)\log(1-\alpha)$ denotes the binary entropy function in nats. Furthermore, for $q>2$, such an $\alpha^*$ always exists, while for $q=2$ such an $\alpha^*$ exists if and only if $C<1$, where $C$ is given by \eqref{eq:C_defn}.
\end{theorem}

\begin{proof}
\label{mindistprf}
Let $P_1(d)$ denote the probability that a length-$N$ vector $\mathbf{c}$ which satisfies $\mathbf{c}\mathbf{H_1}^T = \mathbf{0}$ (i.e., which satisfies the parity checks in the first block row $\mathbf{H}_1$) has Hamming weight $d$. The generating function for this sequence is 	
\begin{equation*}
\label{MGFOnes}
\phi^{(1)}(z) = \sum_{d=0}^{N} P_{1}(d) z^{d} = \prod_{t\in I_c}[\varphi^{(t)}(z)]^{\gamma_t m / q}
\end{equation*}
where $\varphi^{(t)}(z) = A^{(t)}(z) / 2^{k_t}$ is the moment generating function of the Hamming weight of a codeword in $\code_t$. 

Since $P_1(d)\geq 0 \ \forall d$, we can write (for any $z>0$)
\begin{equation}
\label{Pidbound}
P_1(d) \leq \exp\left(\frac{m}{q}\sum_{t \in I_c}\gamma_t\log[\varphi^{(t)}(z)] - d \log z\right).
\end{equation}

Next, let $F_1(d)$ denote the probability that a length-$N$ vector of Hamming weight $d$ satisfies the parity checks in the first block row $\mathbf{H}_1$ of $\mathbf{H}$. An upper bound on this probability is readily deduced from \eqref{Pidbound} as
\begin{align*}
\label{Fidbound}
F_1(d)\!& \leq \frac{2^{m\sum_{t}k_t\gamma_t / q}}{\binom{N}{d}}\exp\left(\frac{m}{q}\sum_{t \in I_c}\gamma_t\log[\varphi^{(t)}(z)] - d \log z\right) \\
&\!= \binom{N}{d}^{-1} \exp\left(\frac{m}{q}\sum_{t \in I_c}\gamma_t\log[A^{(t)}(z)] - d \log z\right) \; .
\end{align*}
Any vector which satisfies all $q$ of the block-rows of $\mathbf{H}$ is a valid codeword for the GLDPC code, and satisfaction of the different block rows are independent events. Thus, an upper bound on $F(d)$, the probability that a length-$N$ vector of weight $d$ is a codeword of the GLDPC code, is given by
\begin{align*} 
F(d)&\!= [F_1(d)]^q\\
&\!\leq \binom{N}{d}^{-q} \exp\left(m\sum_{t \in I_c}\gamma_t\log[A^{(t)}(z)]-qd \log z\right) \; .
\end{align*}
	
The expected number of codewords of weight $d$, for a code chosen uniformly at random from Ensemble 1, is thus
\begin{align*}
M(d)&\!= \!\binom{N}{d}F(d)\\
&\!\leq \!\binom{N}{d}^{-(q-1)} \! \exp\left(m\sum_{t \in I_c}\gamma_t\log[A^{(t)}(z)]-qd \log z\right) .
\end{align*}
For any value of $d$, the tightest bound is obtained when $z$ is chosen to minimize the exponent; this leads to $\sff(z) = \alpha$, where $\sff$ is given by \eqref{eq:f}, and where we let $\alpha = d/N$.

We use this result to bound the probability of the event $d_{\mathrm{min}}\leq d_0$. Using Markov's inequality,
\begin{align*}
\mathrm{Pr}(d_\mathrm{min}\leq d_0)&\leq\sum_{d=1}^{d_0}M(d)\\
&\leq d_0\max_{1\leq d\leq d_0}\left\{ \binom{N}{d}^{-(q-1)} \right.\\
&\quad\left. \! \! \! \times \exp\left(m\sum_{t=1}^{n_c}\gamma_t\log[A^{(t)}(z)]-qd \log z\right)\right\}.
\end{align*}

Using the relation \cite{Gallager1968}
\begin{equation*}
\label{Gal68formula}
\binom{N}{d}\geq \sqrt{\frac{N}{8 d(N-d)}}\exp\left(N h\left(\frac{d}{N}\right)\right)
\end{equation*}
and again using the substitution $\alpha = d/N$, leads to 
\begin{equation}
\label{Pdmin_leq_d0}
\mathrm{Pr}(d_\mathrm{min}\leq d_0)\leq \max_{1\leq d \leq d_0}\exp[NG(\alpha)+o(N)]
\end{equation}
where $G(\alpha)$ is given by \eqref{eq:G(alpha)_irregular_CN_set}, and we have used the fact that the total number of edges in the Tanner graph is $E = m / \int \rho = Nq$. 
Here
\begin{equation*}
\label{oN}
o(N)=\log\left[d_0\left(8N\alpha(1 - \alpha)\right)^{\frac{q-1}{2}}\right] \; .
\end{equation*}
Note that $o(N)/N\rightarrow 0$ as $N\rightarrow\infty$.


To prove that $\mathrm{Pr}(d_\mathrm{min}\leq \alpha_0 N) \rightarrow 0$ as $N \rightarrow \infty$ for a given $\alpha_0$, it suffices to show that $G(\alpha)<0$ for $0<\alpha<\alpha_0$; in this case, the bound in \eqref{Pdmin_leq_d0} guarantees a decrease of $\mathrm{Pr}(d_\mathrm{min}\leq \alpha_0 N)$ to zero with increasing $N$. Next, note that $G(\alpha)$ is identical to the growth rate of the weight distribution as defined at the beginning of \cite[Section III]{Flanagan2012} in the context of the GLDPC ensemble considered therein (to see this connection explicitly, see in particular Definition 4.1 and Theorem 4.2 in \cite{Flanagan2012}). Also, the smallest positive value $\alpha^*$ which solves $G(\alpha) = 0$ exactly matches the value of the \emph{critical exponent codeword weight ratio} of this ensemble \cite[Section III]{Flanagan2012}. From the analysis of the function $G(\alpha)$ previously conducted in \cite{Flanagan2011}, we know that $G(\alpha)<0$ for $0<\alpha<\alpha^*$ if and only if either $q>2$ or $q=2$ and $C<1$.
\end{proof}

\medskip
It is interesting to note that, while the ensemble considered in this latter result is also VN-regular, and the ensemble definitions match from the CN side, they are in fact slightly different ensembles. However, the above result shows that their respective weight distribution behaviors are tightly connected.

\medskip
\begin{example}
\label{mindistexample}
\begin{figure} 
\includegraphics[width=\linewidth]{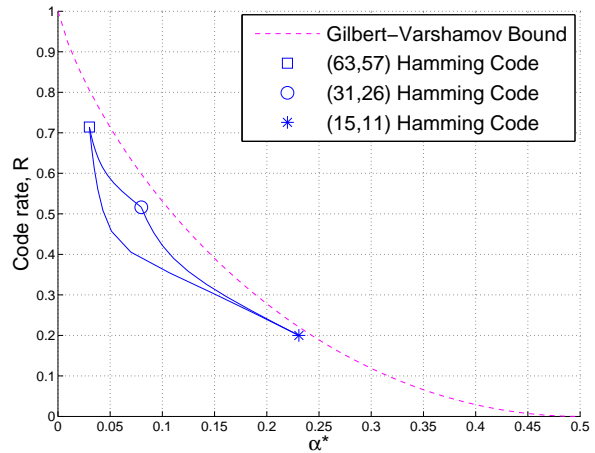}
\caption{Ratios of minimum distance to block length for irregular GLDPC codes, plotted against the design rate $R$ of the ensemble.}
\label{mindistfig}
\end{figure}
In this example, the ensemble relative minimum distance $\alpha^*$ for some irregular GLDPC code ensembles of ensemble type 1 are evaluated using Theorem \ref{mindistthm1}, and plotted against the design rate of the ensemble (as given by \eqref{GLDPC_code_rate_eqn}) -- these results are plotted in Figure~\ref{mindistfig}. Here we use Hamming $(63,57)$, $(31,26)$ and $(15,11)$ codes as the local codes at the CNs. Note that for Hamming codes of length $s_t$, we have
\begin{align*}
A^{(t)}(z) &= \frac{1}{(s_t+1)}\Big((1+z)^{s_t}\\
&\qquad \left.+s_t(1+z)^{(s_t-1)/2}(1-z)^{(s_t+1)/2}\right).
\end{align*}
In each case, two of the three code types are used ($I_c = \{1,2\}$), with $\gamma_1$ being varied between 0 and 1 - this results in the blue curves joining pairs of points corresponding to \emph{regular} GLDPC code ensembles. Note that while mixing CN types tends to bring us further from the Gilbert-Varshamov bound, the threshold of the ensemble will be optimized for some mixture of CN types.
\end{example}

\section{Minimum distance results for Ensemble 2}
In this section, we analyze the minimum distance distribution of Ensemble 2. We will assume throughout this section that the GLDPC ensemble under consideration has at least one CN type with minimum distance equal to $2$. In this context, we will derive an upper bound on $\mathrm{Pr}(d_{\mathrm{min}}\leq\alpha^*N)$ -- note by contrast that for Ensemble 1, Theorem \ref{mindistthm1} represents a related but stronger result, as it defines precisely the conditions under which this value tends to zero as $N \rightarrow \infty$. The results of this section rely on the following basic Lemma.

\medskip
\begin{lemma}
\label{Coef}
For any positive integer $j$,
\begin{equation}
\lim_{N\rightarrow\infty}\mathrm{Coef}\left[\prod_{t\in I_c}\left[A^{(t)}(x)\right]^{\gamma_t m},x^{2j}\right]=\frac{(EC/2)^j}{j!}
\end{equation}
where $\mathrm{Coef}\left[f(x),x^i\right]$ denotes the coefficient of $x^i$ in the Taylor expansion of $f(x)$ (as in \cite{Orlitsky2005}), and the parameter $C$ is given by \eqref{eq:C_defn}.
\end{lemma}

\medskip
\begin{proof}
The proof is notationally cumbersome, so we present it for the example where $I_c = 2$, with $A^{(1)}(x) = 1+A_2x^2+A_3x^3$ and $A^{(2)}(x) = 1+B_2x^2+B_4x^4$. Then
\begin{equation*}
P(x)=(1+A_2x^2+A_3x^3)^{\gamma_1m}(1+B_2x^2+B_4x^4)^{\gamma_2m}.
\end{equation*}
The multinomial theorem then gives
\begin{align*}
P(x)=\sum_{\substack{i_2+i_3\leq\gamma_1m\\j_2+j_4\leq\gamma_2m}} & \binom{\gamma_1m}{i_2 \; i_3}\binom{\gamma_2m}{j_2 \; j_4} \\ 
& \times A_2^{i_2}A_3^{i_3}B_2^{j_2}B_4^{j_4}x^{2(i_2+j_2)+3i_3+4j_4	} \, .
\end{align*}
The coefficient of $x^{2j}$ in $P(x)$ is then given by
\begin{align}
\mathrm{Coef}\left[P(x),x^{2j}\right] = & \sum_{\substack{i_2,i_3,j_2,j_4\\2(i_2+j_2)+3i_3+4j_4 = 2j}} \binom{\gamma_1m}{i_2 \; i_3}\binom{\gamma_2m}{j_2 \; j_4} \notag \\  
& \phantom{-------} \times A_2^{i_2}A_3^{i_3}B_2^{j_2}B_4^{j_4} \; .
\label{eq:coeff_as_sum}
\end{align}
Consider first the sum of all terms where $i_3=j_4=0$. This is	
\begin{align*}
S_1 = &\sum_{i_2+j_2=j}\binom{\gamma_1m}{i_2}\binom{\gamma_2m}{j_2}A_2^{i_2}B_2^{j_2}\\
=&\sum_{i_2=0}^j\binom{\gamma_1m}{i_2}\binom{\gamma_2m}{j-i_2}A_2^{i_2}B_2^{j-i_2}.
\end{align*}
	
As $N\rightarrow\infty$ we have
\begin{align}
S_1 \rightarrow&\sum_{i_2=0}^j\frac{(\gamma_1m)^{i_2}}{i_2!}\frac{(\gamma_2m)^{j-i_2}}{(j-i_2)!}A_2^{i_2}B_2^{j-i_2}\frac{j!}{j!}\nonumber\\
=&\frac{1}{j!}\sum_{i_2=0}^j\binom{j}{i_2}(\gamma_1mA_2)^{i_2}(\gamma_2mB_2)^{j-i_2}\nonumber\\
=&\frac{1}{j!}(\gamma_1mA_2+\gamma_2mB_2)^j\label{star}
\end{align}
	
Note that this term is $\Theta(m^j)$ as $N\rightarrow\infty$. In general, the $(i_2,j_2,i_3,j_4)$ term is
\begin{eqnarray*}
\Theta((\gamma_1m)^{i_2+i_3}(\gamma_2m)^{j_2+j_4}) & = & \Theta( m^{i_2+i_3+j_2+j_4})\\
& = & \Theta(m^\kappa)
\end{eqnarray*}
where the exponent $\kappa$ satisfies
\begin{align*}
\kappa=\frac{2(i_2+i_3+j_2+j_4)}{2}\leq& \frac{2(i_2+j_2)+3i_3+4j_4}{2}\\
\leq& j
\end{align*}
and we conclude that $\kappa \le j$, with strict inequality unless $i_3=j_4=0$.

Since $m^j$ terms dominate all terms $m^\kappa$ with $\kappa < j$, the limiting expression for \eqref{eq:coeff_as_sum} as $N\rightarrow \infty$ involves only those product terms in $P(x)$ for which $r_t = 2$. Therefore, in general we obtain
\begin{align*}
\lim_{N\rightarrow\infty}\mathrm{Coef}& \left[\prod_{t\in I_c}\left[A^{(t)}(x)\right]^{\gamma_tm},x^{2j}\right] \! = \! \frac{1}{j!}\left(\sum_{t:r_t = 2}\gamma_t mA_2^{(t)}\right)^j\nonumber\\
=& \frac{1}{j!}\left[\sum_{t:r_t = 2}\frac{\rho_t m}{s_t\int\rho}A_t^{(t)}\right]^j = \frac{\left(\frac{EC}{2}\right)^j}{j!}.
\end{align*}
\end{proof}

Next we use this result to generalize \cite[Lemma 9]{Orlitsky2005} as follows. 

\medskip
\begin{theorem}
For GLDPC code Ensemble 2, we have	
\begin{equation*}
\label{GenOrl}
\lim_{N\rightarrow\infty}\mathrm{Pr}(d_{\mathrm{min}}=1) = 1-\exp\left(-\frac{\lambda'(0)C}{2}\right) > 0.
\end{equation*}
\end{theorem}

\begin{proof}
We restrict ourselves to considering only degree-2 variable nodes (this may be justified in a manner similar to that described in the proof of \cite[lemma 9]{Orlitsky2005}). Recall that there are $\tilde{\lambda}_2 n$ of these VNs.

From Lemma \ref{Coef} in the special case $j=1$, we have
\begin{equation*}
\lim_{N\rightarrow\infty}\mathrm{Coef}\left[\prod_{t\in I_c}\left[A^{(t)}(x)\right]^{\gamma_t m},x^2\right]=\frac{EC}{2}.
\end{equation*}
Let $A_i$ denote the event that VN $\{v_i\}$ is a codeword (of Hamming weight $1$) of the GLDPC code:
\begin{align*}
A_i=&\left\{\left\{v_i\right\}\mathrm{\ is \ a \ codeword}\right\}\nonumber \\
=& \left\{\left\{v_i\right\}\in C\right\},\nonumber 
\end{align*}
Then $\mathrm{Pr}(d_{\mathrm{min}}=1)$ may be written as a union of such events, which may then be expanded using the inclusion-exclusion principle:
\begin{align}
&\mathrm{Pr}(d_{\mathrm{min}}=1)=\mathrm{Pr}\left(\cup_{i}A_i\right)\nonumber\\
&= \sum_{1\leq i_1\leq\tilde{\lambda}_2 n}\mathrm{Pr}(A_{i_1})-\sum_{1\leq i_1,i_2\leq \tilde{\lambda}_2 n}\mathrm{Pr}(A_{i_1},A_{i_2}) + \cdots \label{eq:alternating_sum}
\end{align}
The general term in this alternating sum is the sum, evaluated over all sets $\mathcal{V}_2 = \{v_{i_1}, v_{i_2}, \dots, v_{i_j}\}$ of $j$ degree-$2$ VNs, of the probability that all VNs in the set $\mathcal{V}_2$ individually form codewords (of Hamming weight $1$), i.e.,
\begin{align*}
&\sum_{1\leq i_1<i_2<\dots<i_j\leq\tilde{\lambda}_2n}\mathrm{Pr}\left(\left\{v_{i_1}\right\},\left\{v_{i_2}\right\},\dots,\left\{v_{i_j}\right\}\in C\right)\\
&\qquad\qquad = \binom{\tilde{\lambda}_2n}{j}\frac{\left\{\mathrm{Coeff}\left[\prod_{t\in I_c}\left[A^{(t)}(x)\right]^{\gamma_t m},x^2\right]\right\}^j}{\binom{E}{2 \; 2 \; 2 \cdots \; 2}}.
\end{align*}
In the fraction above, the denominator is the number of ways of choosing $j$ pairs of edges in the Tanner graph, while the numerator counts the number of these such that each pair individually satisfies the CN constraints (i.e., when 1s are placed on this pair of edges, and zeros on the other edges). In the limit as $N\rightarrow\infty$ we obtain (invoking Lemma \ref{Coef})
\begin{align*}
& \lim_{N\rightarrow\infty}\sum_{1\leq i_1<i_2<\dots<i_j\leq\tilde{\lambda}_2n}\mathrm{Pr}\left(\left\{v_{i_1}\right\},\left\{v_{i_2}\right\},\dots,\left\{v_{i_j}\right\}\in C\right)\\
&\qquad \qquad= \frac{\left(\frac{\lambda'(0)E}{2}\right)^j}{j!}\frac{\left(\frac{EC}{2}\right)^j}{E^{2j}} 2^j \nonumber\\
&\qquad\qquad=\frac{\left[\frac{\lambda'(0)C}{2}\right]^j}{j!} \; ,\nonumber
\end{align*}
where we have made use of the fact that
\begin{equation}
\tilde{\lambda}_2n = \frac{\lambda_2 E}{2} = \frac{\lambda'(0)E}{2} \; .
\label{eq:lambda_prime_lambda_tilde}
\end{equation}
Substituting this result into \eqref{eq:alternating_sum} and using the Taylor series of the exponential function yields the result of the Theorem.
\end{proof}
\medskip

This result shows that codes from Ensemble 2 (assuming the existence of CN types of minimum distance $2$) have a nonzero probability of having minimum distance equal to $1$. Note that such codes can be removed by ensemble expurgation.

Next we prove an upper bound on the probability of the minimum distance for a GLDPC code, which generalizes \cite[Lemma 22]{Orlitsky2005}.

\medskip
\begin{theorem}
\label{mindistthm2}
For GLDPC code Ensemble 2,
\begin{equation*}
\mathrm{Pr}(d_{\mathrm{min}}\leq \alpha^*N)\leq \frac{1}{\sqrt{1-\lambda'(0)C}}-1.
\end{equation*}	
\end{theorem}

\begin{proof}
We denote by $\mathbb{V}$ and $\mathbb{C}$ the set of length-$N$ binary vectors and the set of codewords, respectively.
Consider events $\{S\in \mathbb{C}\}$, where $S \in \mathbb{V}$. The probability that $d_\mathrm{min}\leq \alpha^*N$ for a code is equal to the probability that for all possible $S$ with $|S|\leq \alpha^* N$, \emph{at least} one of them is a member of $\mathbb{C}$, which in turn is equal to the probability of the union of events $\{S\in \mathbb{C}\}$ where $|S|\leq \alpha^* N$, i.e.,
\begin{equation*}
\mathrm{Pr}(d_{\mathrm{min}}\leq\alpha^*N)=\mathrm{Pr}\left[\bigcup_{S,|S|\leq\alpha^*N}\{S\in \mathbb{C}\}\right].
\end{equation*}

Using the union bound,
\begin{align}
\label{PrUnionBd}
\mathrm{Pr}\left[\bigcup_{S,|S|\leq\alpha^*N}\{S\in \mathbb{C}\}\right]
&\leq \sum_{S,|S|\leq\alpha^*N}\mathrm{Pr}(S\in \mathbb{C})\nonumber\\
&= \sum_{j=1}^{\alpha^*N}\left[\sum_{S,|S|=j}\mathrm{Pr}(S\in \mathbb{C})\right].
\end{align}
Since all of the summands in the innermost summation in (\ref{PrUnionBd}) are equal, we have	
\begin{equation*}
\sum_{S,|S|=j}\mathrm{Pr}(S\in \mathbb{C}) \! = \! \binom{\tilde{\lambda}_2n}{j}\frac{\mathrm{Coef}\left[\prod_{t\in I_c}\left[A^{(t)}(x)\right]^{\gamma_t m},x^{2j}\right]}{\binom{E}{2j}}.
\end{equation*}
In the fraction above, the denominator is the number of ways of choosing $2j$ edges in the Tanner graph, while the numerator counts the number of these such that placing 1s on these edges and 0s on the other edges, satisfies all of the CN constraints. 
Taking the limit as $N\rightarrow\infty$ and invoking Lemma \ref{Coef}, we obtain
\begin{align*}		
\lim_{N\rightarrow\infty}\sum_{S,|S|=j}\mathrm{Pr}(S\in \mathbb{C}) &= \frac{\left[\frac{\lambda'(0)E}{2}\right]^j}{j!}\frac{\left[\frac{EC}{2}\right]^j}{j!}\frac{(2j)!}{E^{2j}}\\
&= \binom{2j}{j} \left(\frac{\lambda'(0)C}{4}\right)^j \; ,
\end{align*}
where we have again used \eqref{eq:lambda_prime_lambda_tilde}. Therefore,
\begin{equation}
\label{sumlambda}
\mathrm{Pr}(d_{\mathrm{min}}\leq \alpha^*N)\leq \sum_{j=1}^\infty \binom{2j}{j} \left(\frac{\lambda'(0)C}{4}\right)^j.
\end{equation}
The generating function for the central binomial coefficient is given by \cite{Wilf1994}
\begin{equation}
\label{cbcgen}	
\sum_{j=0}^{\infty}\binom{2j}{j}x^j=\frac{1}{\sqrt{1-4x}}.
\end{equation}
Finally, inserting (\ref{cbcgen}) into (\ref{sumlambda}) we arrive at the statement of the theorem.	
\end{proof}

\bibliographystyle{IEEEtran}
\bibliography{IEEEabrv,main}

\end{document}